\documentclass[conference]{IEEEtran}
\IEEEoverridecommandlockouts
\usepackage{cite}
\usepackage{amsmath,amssymb,amsfonts}
\usepackage{graphicx}
\usepackage{textcomp}
\usepackage{xcolor}
\usepackage{stfloats}
\usepackage{algpseudocode}
\usepackage{algorithm}
\usepackage{verbatim}

\newtheorem{proof}{Proof}
\newtheorem{Proposition}{Proposition}
\def\BibTeX{{\rm B\kern-.05em{\sc i\kern-.025em b}\kern-.08em
    T\kern-.1667em\lower.7ex\hbox{E}\kern-.125emX}}
\begin{document}

\title{MIMO Transmission Under Discrete \\Input Signal Constraints
{}
\thanks{J. Feng, B. Feng, Y. Wu, S. Li, and W. Zhang are with the Department of Electronic Engineering, Shanghai Jiao Tong University, Minhang 200240, China (e-mail: search4meaning@sjtu.edu.cn; fengbiqian@sjtu.edu.cn; yongpeng.wu@sjtu.edu.cn; shen-l@sjtu.edu.cn; zhangwenjun@sjtu.edu.cn)(Corresponding author: Yongpeng Wu.).}
}

\author{\IEEEauthorblockN{Jie Feng, Biqian Feng, Yongpeng Wu, Li Shen, Wenjun Zhang}
}
\maketitle

\begin{abstract}
In this paper, we propose a multiple-input multiple-output (MIMO) transmission strategy that is closer to the Shannon limit than the existing strategies. Different from most existing strategies which only consider uniformly distributed discrete input signals, we present a unified framework to optimize the MIMO precoder and the discrete input signal distribution jointly. First, a general model of MIMO transmission under discrete input signals and its equivalent formulation are established. Next, in order to maximize the mutual information between the input and output signals, we provide an algorithm that jointly optimizes the precoder and the input distribution. Finally, we compare our strategy with other existing strategies in the simulation. Numerical results indicate that our strategy narrows the gap between the mutual information and Shannon limit, and shows a lower frame error rate in simulation.

\end{abstract}

\begin{IEEEkeywords}
Discrete input signals, MIMO, Shannon limit
\end{IEEEkeywords}

\vspace{1.5 ex}
\section{Introduction}

\subsection{Background}
The Shannon limit \cite{Shannon} theoretically specifies the maximum rate of reliable transmission with given power over the additive white Gaussian noise (AWGN) channel. And researchers are seeking various methods to approach the Shannon limit under discrete input signal constraints to improve the spectrum efficiency in practice. However, most researches on increasing the capacity of MIMO communication systems are based on the uniformly distributed discrete input constellation, remaining a gap between the capacity-achieving input which requires Gaussian distribution.

\subsection{Related work}

In single-input single-output (SISO) communication systems, probabilistic shaping methods which can convert the discrete input into a more Gaussian-like distribution have been widely studied. However, many of these shaping schemes [2-4] have the common shortcoming of codec complexity and rate inflexibility. In \cite{PAS}, the probabilistic amplitude shaping (PAS) with low complexity and significant shaping gain is designed to change the distribution of amplitude shift keying (ASK) or quadrature amplitude modulation (QAM) inputs. At the transmitter, a Constant Composition Distribution Matcher (CCDM) \cite{CCDM} and a systematic binary encoder using DVB-S2 low-density parity-check (LDPC) perform PAS and channel coding. At a frame error rate (FER) of $10^{-3}$, the PAS scheme operates within less than 1 dB of the AWGN capacity at any rate between 1 and 5 bits per channel use. 

In MIMO communication systems, over the last decade, most works [7-11] focus on increasing the MIMO capacity by linear precoding under the premise that the discrete input signals are uniformly distributed, leading to a gap from Shannon limit which requires Gaussian inputs \cite{Telatar}. In \cite{Xiao}, globally optimal linear precoders for uniformly distributed discrete signals over complex Gaussian channels are proposed. The main idea of designing the optimal precoder is to jointly optimize a power allocation matrix and a unitary matrix over the equivalent parallelized channel model. Though the optimal precoder proposed in \cite{Xiao} greatly improves the system capacity, it still ignores the limiting factor of the input signals being non-Gaussian distributed. Thus, we  see the potential of applying discrete input signals with non-uniform distributions to improve MIMO capacity.


\newcounter{mytempeqncnt}
\begin{figure*}[!t]
 \normalsize
 \setcounter{mytempeqncnt}{\value{equation}}
 \setcounter{equation}{4}
 \begin{subequations}
 \label{I}
 \begin{align}
\mathcal{I}\mathbf{(x;\bar{y})}&=-\sum_{i=1}^{M^{N_t}}p(\mathbf{x}_i)\left ( \log_{2}{p(\mathbf{x}_i)}
+E_\mathbf{v}\left [ \log_{2}{\left ( \sum_{p=1}^{M^{N_t}}\frac{p(\mathbf{x}_p)}{p(\mathbf{x}_i)}e^{-a_{i,p}} \right ) } \right ] \right ),\\
a_{i,p}&= \sigma^{-2}(\left \| \mathbf{\Sigma_H\Sigma_G\Phi\Delta(x_\mathit{i}-x_\mathit{p})+v} \right \|^2 -\left \| \mathbf{v} \right \|^2 ).
 \end{align}
 \end{subequations}
 \setcounter{equation}{5}
 \hrulefill
 \vspace*{4pt}
\end{figure*}

\subsection{Our work}
The intention of our work is to present a unified framework of MIMO transmission with linear precoders and non-uniformly distributed discrete input signals. Comparing to SISO case, the distribution optimization becomes more complicated with vector input signals. This is because the input distribution of any antenna will affect the optimization of others. Besides, the joint design of the precoder and the input distribution poses another difficulty in the optimization problem.

We first define the original system model and its equivalent form. We use the maximum mutual information criterion to increase the system capacity to further approach the Shannon limit. We formulate the optimization problem and the variables to be optimized including a power allocation matrix, a unitary matrix and the input distribution.

Secondly, the gradient descent algorithm, the manifold optimization, and the coordinate descent algorithm are applied to jointly solve the problem.

Next, we provide simulations to evaluate the transmission strategy based on the proposed framework. Numerical results show that the proposed strategy outperforms the strategy for uniform discrete input signals in \cite{Xiao} over both constant and Rayleigh fading channels. Moreover, the signal-noise-ratio (SNR) gain for FER becomes more obvious for MIMO systems employing minimum mean mquared error (MMSE) detector.

In general, the contributions of our work are summarized as follows:

\begin{itemize}
\item The optimal distribution for a MIMO system with non-uniformly discrete input signals is studied. 
\item A unified framework for a joint MIMO precoder and input distribution optimization is established. 
\item Numerical results indicates that the proposed design outperforms the existing design for uniformly discrete input signals in terms of both mutual information and FER.
\end{itemize}

The rest of this paper is organized as follows. In Section \uppercase\expandafter{\romannumeral2}, we propose our system model. In Section \uppercase\expandafter{\romannumeral3}, we formulate the mutual information maximization problem and provide an algorithm jointly optimizes the variables. Numerical results are provided in Section \uppercase\expandafter{\romannumeral4}. Finally, conclusions are drawn in Section \uppercase\expandafter{\romannumeral5}.

\emph{Notation:} Lowercase boldface letters $\mathbf{a}$ and uppercase boldface letters $\mathbf{A}$ respectively denote column vectors and matrices, $\mathbf{I}_N$ denotes an $N\times N$ identity matrix, diag$(\mathbf{a})$ or diag$(\mathbf{a}^h)$ denotes a diagonal matrix with diagonal entries formed by $\mathbf{a}$, diag$(\mathbf{A})$ denotes a diagonal matrix containing the diagonal of matrix $\mathbf{A}$, $\mathrm{det}(\cdot)$ denotes the matrix determinant, $\mathrm{trace}(\cdot)$ denotes the trace operation, $E_\mathbf{v}[\cdot]$ represents the expectation with respect to random variable $\mathbf{v}$, the Euclidean norm operator is denoted by $\left \| \cdot  \right \|$ and the superscripts $(\cdot)^t$, $(\cdot)^\ast$ and  $(\cdot)^h$ represent transpose, conjugate, and conjugate transpose operations, respectively.

\vspace{1.5 ex}
\section{System Model}
Consider a single-user MIMO Gaussian channel where the transmitter and the receiver are respectively equipped with $N_t$ and $N_r$ antennas. The input-output relationship can be written as
\begin{equation}
\mathbf{y=H G \Delta x+v},\tag{1}
\label{model_1}
\end{equation}
where $\mathbf{y}\in\mathbb{C} ^{N_r\times 1}$ is the received channel output signal, $\mathbf{H}\in\mathbb{C} ^{N_r\times N_t}$ is a random channel matrix denoting the complex fading coefficient of each pair of the transmit and receive antenna, $\mathbf{G}\in\mathbb{C} ^{N_t\times N_t}$ is the linear precoder, and $\mathbf{v}\in\mathbb{C} ^{N_r\times 1}$ is a zero-mean complex Gaussian noise vector with covariance $\sigma^2\mathbf{I}_{N_r}$. Let $\mathbf{x}=\left ( x_1, \dots , x_{N_t} \right ) ^t$ denote the modulated symbol vector. Under the $M$-QAM ($M=2^m$) modulation scheme, for $i$th antenna, the in-phase and quadrature elements of $x_i$ are chosen from the alphabet $\chi=\left \{ \pm 1, \pm3, \dots, \pm\sqrt{M}-1 \right \}$. The constellation scaling matrix $\mathbf{\Delta}=\mathrm{diag}\left ( \Delta_1, \dots , \Delta_{N_t} \right ) $ is applied to satisfy the unit power constraint $E_{\mathbf{x} }[\mathbf{\Delta x (\Delta x)}^h]=\mathbf{I}_{N_t}$. And the overall transmitting power constraint is satisfied by $\mathrm{trace}(\mathbf{GG^\mathit{h}})=P$.

Consider a deterministic channel $\mathbf{H}$ with SVD decomposition $\mathbf{H}=\mathbf{U_H \Sigma _H V_H^\mathit{h}}$ which is perfectly known at the transceiver. From \cite[Prop. 2]{opU}, when designing the optimal precoder with SVD decomposition $\mathbf{G}=\mathbf{U_G \Sigma _G \Phi}$, for any input distribution of $\mathbf{\Delta x}$, the left singular matrix of $\mathbf{G}$ can always be chosen to be the right singular vectors of $\mathbf{H}$, i.e., $\mathbf{U_G}=\mathbf{V_H}$. Then, based on \cite[Prop. 1]{Xiao}, let $\mathbf{\bar{y}}=\mathbf{U_H^\mathit{h}}\mathbf{y}$ and model \eqref{model_1} can be simplified as the following equivalent model with $\mathcal{I}\mathbf{(x; \bar{y})}= \mathcal{I}\mathbf{(x; y)}$
\begin{equation}
\mathbf{\bar{y}= \Sigma _H \Sigma _G \Phi \Delta x + v},\tag{2}
\label{model_2}
\end{equation}
where $\mathbf{\Sigma_H}$ and $\mathbf{\Sigma_G}$ are diagonal matrices with nonnegative elements, $\mathbf{\Phi}$ is a unitary matrix and $\mathbf{v}$ is the noise vector with its statistical information unchanged.

To find the optimal MIMO transmitting strategy, we aim to maximize the mutual information $\mathcal{I} \mathbf{(x; \bar{y})}$ which is positively correlated with FER performance. For given $\mathbf{\Sigma _H}$ and $ \mathbf{\Sigma _G}$, the probability density function of $\mathbf{\bar{y}}$ can be computed by
\begin{equation}
p(\mathbf{\bar{y}})=E_\mathbf{x}[p(\mathbf{\bar{y}|x})]=\sum_{i=1}^{M^{N_t}}p(\mathbf{x_\mathit{i}})p(\mathbf{\bar{y}|x_\mathit{i}}) ,\tag{3}
\label{py}
\end{equation}
\begin{equation}
p(\mathbf{\bar{y}|x})=\frac{1}{(\pi\sigma^2)^{N_r}}\mathrm{exp}\left(\frac{\left \|  \mathbf{\bar{y}-\Sigma_H\Sigma_G\Phi \Delta x}  \right \|^2 }{\sigma^2}\right).\tag{4}
\label{pyx}
\end{equation}
Then, $\mathcal{I} \mathbf{(x; \bar{y})}$ is given by \eqref{I}, where the probability density function of $\mathbf{v}$ is given by

\begin{equation}
p(\mathbf{v})=\frac{1}{(\pi\sigma^2)^{N_r}}\exp\left ( -\frac{\left \| \mathbf{v}\right \|^2 }{\sigma^2} \right ).
\label{pv}
\end{equation}

After presenting the framework of MIMO transmission with non-uniformly distributed input signals, we aim to design an effective strategy to maximize the mutual information. The corresponding optimization problem is formulated as:
	\begin{subequations}
	\begin{eqnarray}
	\underset{\mathbf{\Sigma}_{\mathbf{G}},\mathbf{\Phi},p\left(\mathbf x\right)}{\max} &&\mathcal I\left(\mathbf x;\bar{\mathbf y}\right)\\
	\operatorname{s.t.}&&\mathbf{\Phi}\mathbf{\Phi}^h=\mathbf I_{N_t},\label{constraintPhi}\\
	&&\operatorname{trace}\left(\mathbf{\Sigma}_{\mathbf{G}}\mathbf{\Sigma}_{\mathbf{G}}^h\right)=P, \mathbf{\Sigma_G \succeq 0},\label{constraintSigmaG}\\
	&&E_{\mathbf x}\left[\mathbf \Delta\mathbf x\left(\mathbf \Delta\mathbf x\right)^h\right]=\mathbf I_{N_t}.\label{constraintDelta1}
	\end{eqnarray}
	\end{subequations}

\begin{figure*}[!t]
 \normalsize
 \setcounter{mytempeqncnt}{\value{equation}}
 \setcounter{equation}{7}
 \begin{subequations}
 \label{GDG}
 \begin{align}
\frac{\partial\mathcal{I}\mathbf{(x;\bar{y})}}{\partial \mathbf{{\Sigma_G}^\ast }}&=\mathrm{diag}\left ( \sum_{i=1}^{M^{N_t}}p(\mathbf{x}_i)E_\mathbf{v}\left [  \frac{\sum_{p=1}^{M^{N_t}}{p(\mathbf{x}_p)e^{-b_{i,p}}c_{i,p}}}{\sum_{p=1}^{M^{N_t}}{p(\mathbf{x}_p)e^{-b_{i,p}}}}\right ] \right ),\\
b_{i,p}&=\sigma^{-2}\left \|\mathbf{ \Sigma_H \Sigma_G\Phi\Delta(x_\mathit{i}-x_\mathit{p})+v }\right \|^2,\\
c_{i,p}&=\sigma^{-2}\mathbf{(\Sigma_H^2 \Sigma_G\Phi\Delta(x_\mathit{i}-x_\mathit{p})(x_\mathit{i}-x_\mathit{p})^\mathit{h}\Delta\Phi^{\mathit{h}}+\Sigma_Hv (x_\mathit{i}-x_\mathit{p})^\mathit{h}\Delta\Phi^{\mathit{h}})}.
 \end{align}
 \end{subequations}
 \setcounter{equation}{8}
 \hrulefill
\end{figure*}

For convenience, we use the notations of the input distribution $p\mathbf{(x)}$ and the constellation scaling matrix $\mathbf{\Delta}$ interchangeably since they are determined by each other in (\ref{constraintDelta1}). Obviously, there exists a feasible point $\mathbf{\Sigma}_{\mathbf{G}}=\sqrt{\frac{P}{N_t}}\mathbf I_{N_t}, \mathbf \Phi=\mathbf I_{N_t}, \mathbf \Delta=\sqrt{\frac{3}{2\left(M-1\right)}}\mathbf I_{N_t}$ satisfying the constraints (\ref{constraintPhi})-\eqref{constraintDelta1}, which verifies the feasible of the problem \cite[\S 4.1.1]{ConvexOptimization}. However, it is worth noting that the optimization problem is quite challenging due to the following reasons. First, computing the objective function of $\mathcal I\left(\mathbf x;\bar{\mathbf y}\right)$ with expectation term is computationally expensive, and so is the constraint \eqref{constraintDelta1}. Secondly, the objective function is non-concave with respect to $\mathbf \Sigma_{\mathbf G}$, $\mathbf \Phi$, and $\mathbf \Delta$. Additionally, the presence of the unitary matrix $\mathbf \Phi$ further complicates the optimization procedure.
\vspace{1.5 ex}
\section{Mutual Information Optimization}

In this section, we develop an efficient algorithm for the problem which optimizes the power allocation matrix $\mathbf{\Sigma_G}$,  unitary matrix $\mathbf{\Phi}$,  and the distribution of input signal $\mathbf{\Delta x}$ in an alternating manner. Specifically, we solve it by solving the following three subproblems iteratively: optimize $\mathbf{\Sigma_G}$ with given $\mathbf{\Phi}$ and $\mathbf{\Delta x}$, optimize $\mathbf{\Phi}$ with given $\mathbf{\Sigma_G}$  and $\mathbf{\Delta x}$, and optimize $\mathbf{\Delta x}$ with given $\mathbf{\Sigma_G}$ and $\mathbf{\Phi}$.  Then we present the overall algorithm and show its convergence.

\subsection{Optimize power allocation matrix $\mathbf{\Sigma_G}$}
Though mercury-waterfilling proposed in \cite{MWF} is the optimal power allocation policy for parallel channels, it can not be applied here because of  $\mathbf{\Phi}$.
To handle the problem, we adopt the gradient projection method \cite[\S 3.3.1]{GradientProjection} to optimize $\mathbf \Sigma_{\mathbf G}$ with fixed $\mathbf\Phi$ and $\mathbf \Delta$. The subproblem is given by
	\begin{subequations}
	\begin{eqnarray}
	\underset{\mathbf{\Sigma}_{\mathbf{G}}}{\max} &&\mathcal I\left(\mathbf x;\bar{\mathbf y}\right)\\
	\operatorname{s.t.}&&\operatorname{trace}\left(\mathbf{\Sigma}_{\mathbf{G}}\mathbf{\Sigma}_{\mathbf{G}}^h\right)=P,\mathbf{\Sigma_G \succeq 0}\label{constraintSigmaG1}.
	\end{eqnarray}
	\end{subequations}

 \begin{Proposition}
 The projection of the result $\hat{\mathbf \Sigma}_{\mathbf G}$ onto a set $\mathcal Z\triangleq \left\{\mathbf\Sigma_{\mathbf G}:\operatorname{trace}\left(\mathbf{\Sigma}_{\mathbf{G}}\mathbf{\Sigma}_{\mathbf{G}}^h\right)<  P\right\}$ is given by
 \begin{equation}
 \begin{aligned}
 \mathbf\Pi_{\mathcal Z}[\hat{\mathbf \Sigma}_{\mathbf G}]&=\underset{\mathbf  \Sigma_{\mathbf G}\in\mathcal Z }{\arg\min}\,\,\|\mathbf \Sigma_{\mathbf G}-\hat{\mathbf \Sigma}_{\mathbf G}\|^2\\
 &=\sqrt{\frac{P}{\operatorname{trace}\left[ \left[\hat{\mathbf\Sigma}_{\mathbf G}\right]^+\left[\hat{\mathbf\Sigma}_{\mathbf G}^h\right]^+ \right]}}\left[\hat{\mathbf\Sigma}_{\mathbf G}\right]^+,
 \end{aligned}
 \label{projection}
 \end{equation}
 where $\left[\mathbf{X}\right]^+$ denotes the projection of $\mathbf{X}$ onto the positive semidefinite cone.
 \end{Proposition}

 \begin{proof}
 In the following we solve the above problem via the partial Lagrangian function.
 \begin{equation}
 \begin{aligned}
 \mathcal{L}\left(\mathbf \Sigma_{\mathbf G},\lambda\right)&\!=\!\|\mathbf \Sigma_{\mathbf G}\!\!-\!\!\hat{\mathbf \Sigma}_{\mathbf G}\|^2+\lambda\left[\operatorname{trace}\left(\mathbf{\Sigma}_{\mathbf{G}}\mathbf{\Sigma}_{\mathbf{G}}^h\right)\!-\!P\right].
 \end{aligned}
 \end{equation}
 The dual function is given by $g(\lambda)=\mathop{\inf}\limits_{\mathbf \Sigma_{\mathbf G}}\mathcal{L}\left(\mathbf \Sigma_{\mathbf G},\lambda\right)$. Since $\mathcal{L}\left(\mathbf \Sigma_{\mathbf G},\lambda\right)$ is a convex function of $\mathbf \Sigma_{\mathbf G}$, we can find the optimal matrices $\mathbf \Sigma_{\mathbf G}$ from the optimality condition
 \begin{equation}
 \begin{aligned}
 \nabla_{\mathbf \Sigma_{\mathbf G}^*} \mathcal{L}\left(\mathbf \Sigma_{\mathbf G},\lambda\right)=\left(1+\lambda\right)\mathbf\Sigma_{\mathbf G}-\hat{\mathbf\Sigma}_{\mathbf G}=\mathbf{0},
 \end{aligned}
 \end{equation}
 which yields $\mathbf\Pi_{\mathcal Z}[\hat{\mathbf \Sigma}_{\mathbf G}]=\frac{1}{1+\lambda}\hat{\mathbf\Sigma}_{\mathbf G}$. Then, $\mathbf\Sigma_{\mathbf G}$ is projected onto the positive semidefinite cone, which leads to the desired $\mathbf\Pi_{\mathcal Z}[\hat{\mathbf \Sigma}_{\mathbf G}]=\frac{1}{1+\lambda}\left[\hat{\mathbf\Sigma}_{\mathbf G}\right]^+$.

 Putting the closed-form solution into the constraint, we have
 \begin{equation}
 \operatorname{trace}\left(\mathbf\Pi_{\mathcal Z}[\hat{\mathbf \Sigma}_{\mathbf G}]\mathbf\Pi_{\mathcal Z}[\hat{\mathbf \Sigma}_{\mathbf G}]^h\right)\!\!=\!\!\frac{\operatorname{trace}\left[\left[\hat{\mathbf\Sigma}_{\mathbf G}\right]^+\!\!\left[\hat{\mathbf\Sigma}_{\mathbf G}^h\right]^+\!\! \right]}{\left(1+\lambda\right)^2}=P.
 \end{equation}
 Therefore,
 \begin{equation} \mathbf\Pi_{\mathcal Z}[\hat{\mathbf \Sigma}_{\mathbf G}]\!\!=\!\!\frac{\left[\hat{\mathbf\Sigma}_{\mathbf G}\right]^+}{1\!\!+\!\!\lambda}\!\!=\!\!\!\!\sqrt{\frac{P}{\operatorname{trace}\left[\ \left[\hat{\mathbf\Sigma}_{\mathbf G}\right]^+ \!\!\left[\hat{\mathbf\Sigma}_{\mathbf G}^h\right]^+ \right]}}\left[\hat{\mathbf\Sigma}_{\mathbf G}\right]^+.
 \end{equation}
 \end{proof}
See complete steps for optimizing $\mathbf{\Sigma_G}$ in \textbf{Algorithm 1}.

\begin{algorithm}[tb]
\caption{The gradient descent algorithm for optimizing $\mathbf{\Sigma_G}$ with given $\mathbf{\Sigma_H, \Phi}$ and $p\mathbf{(x)}$.}
    \begin{algorithmic}[1]
    \State Initialize $k=0$ and $\mathbf{\Sigma_G}_k=\sqrt{\frac{P}{N_t}}\mathbf{I}_{N_t}$.
    \State Compute the gradient $\mathbf{\Gamma}_k$ for $\mathbf{\Sigma_G}$ in (8).
    \State Update $\mathbf{\Sigma_G}_{k+1}$ by $\mathbf{\Sigma_G}_k+\mu[\mathbf{\Gamma}_k-\mathrm{trace}(\mathbf{\Gamma}_k/N_t)]$ with step size $\mu$ determined by the backtracking line search.
    \State Update $\mathbf{\Sigma_G}_{k+1}$ according to  (10). Compute the mutual information $\mathcal{I}_{k+1}$. $k:=k+1$.
    \State Repeat step 2-4 until step size $\mu$ approaches zero.
    \end{algorithmic}
\end{algorithm}

\subsection{Optimize unitary matrix $\mathbf{\Phi}$}

For given $\mathbf \Sigma_{\mathbf G}$ and $\mathbf \Delta$, we can rewrite the initial problem as
	\begin{subequations}

	\begin{eqnarray}
	\underset{\mathbf{\Phi}}{\max} &&\mathcal I\left(\mathbf x;\bar{\mathbf y}\right)\\
	\operatorname{s.t.}&&\mathbf{\Phi}\mathbf{\Phi}^h=\mathbf I_{N_t}.\label{constraintPhi2}
	\end{eqnarray}
	\end{subequations}
	Some relax-based algorithms hardly guarantee a stationary point under unitary matrix constraint. Here, we introduce a manifold optimization algorithm to overcome this drawback \cite{manifold}. The constraint defines an Stiefel manifold which can be characterized by
	\begin{equation}
	\mathcal S=\left\{\mathbf\Phi:\mathbf\Phi\mathbf\Phi^h=\mathbf I_{M_t}\right\}
	\end{equation}

	The tangent space of the Stiefel manifold $\mathcal S$ at each point $\mathbf\Phi$ is identified as the matrix space $\left\{\mathbf X:\mathbf X^h\mathbf\Phi+\mathbf \Phi^h\mathbf X=\mathbf 0\right\}$. In this paper, we set the Riemannian metric as $\langle\mathbf Z_1,\mathbf Z_2\rangle_{\mathbf\Phi}=\frac{1}{2}\mathcal{R}\left\{\operatorname{trace}\left(\mathbf Z_1-\mathbf Z_2\right)\right\}$. The Riemannian gradient at a point $\mathbf \Phi$ is
	\begin{equation}
	\nabla_{\mathbf\Phi^{*}} \mathcal{I}(\mathbf{x} ; \bar{\mathbf{y}})=\mathbf \Gamma_{\mathbf\Phi}-\mathbf\Phi\mathbf \Gamma_{\mathbf\Phi}^h\mathbf\Phi
	\end{equation}
	where $\mathbf \Gamma_{\mathbf\Phi}=\frac{\partial \mathcal{I}(\mathbf{x} ; \bar{\mathbf{y}})}{\partial \mathbf\Phi^{*}}$ is the gradient on the Euclidean space at a given $\mathbf \Phi$.
	
	Furthermore, the unitary optimization can be solved in an iterative manner by using a steepest descent algorithm, and the corresponding rotational update at iteration $k$ is given by:
	\begin{equation}
	\mathbf \Phi_{k+1}=\exp\left(-\mu_k\mathbf G_k\right)\mathbf\Phi_k
	\end{equation}
	where $\mathbf G_k\triangleq\nabla_{\mathbf\Phi^{*}} \mathcal{I}(\mathbf{x} ; \bar{\mathbf{y}})\mathbf W_k^H=\mathbf \Gamma_{\mathbf\Phi}\mathbf\Phi_k^h-\mathbf\Phi_k\mathbf \Gamma_{\mathbf\Phi}^h$. The step size $\mu_k$ controls the convergence speed and needs to be computed at each iteration. Backtrack line search based on Armijo-Goldstein condition is efficiently used \cite{SD}. More details are provided in \textbf{Algorithm 2}.

\begin{algorithm}[tb]
\caption{The SD algorithm on the Riemannian space for optimizing $\mathbf{\Phi}$ with given $\mathbf{\Sigma_H, \Sigma_G}$ and $p\mathbf{(x)}$.}
    \begin{algorithmic}[1]
    \State Initialize $k=0$ and $\mathbf{\Phi}_k=\mathbf{I}_\mathit{N_t}$ .
    \State Compute the gradient $\mathbf{\Gamma}_k$ for $\mathbf{\Phi}$ on the Euclidean space in \eqref{SDV}.
    \State Compute the gradient direction on the Riemannian space: $\mathbf{R_\mathit{k}=\Gamma_\mathit{k}\Phi\mathit{_k^h}-\Phi_\mathit{k}\Gamma\mathit{_k^h}}$.
    \State Determine the rotation matrix: $\mathbf{P}_k=\mathrm{exp}(\mu \mathbf{R}_k)$ with step size $\mu$ determined by the backtracking line search.
    \State Update $\mathbf{\Phi}_{k+1}=\mathbf{P_k \Phi}_k$ and compute the mutual information $\mathcal{I}_{k+1}$. $k:=k+1$.
    \State Repeat step 2-4 until step size $\mu$ approaches zero.
    \end{algorithmic}
\end{algorithm}

\begin{figure*}[!t]
	\normalsize
	\setcounter{mytempeqncnt}{\value{equation}}
	\setcounter{equation}{18}
	\begin{subequations}
		\label{SDV}
		\begin{align}
		\frac{\partial\mathcal{I}\mathbf{(x;\bar{y})}}{\partial \mathbf{\Phi^\ast }}&=\sum_{i=1}^{M^{N_t}}p(\mathbf{x}_i)E_\mathbf{v}\left [  \frac{\sum_{p=1}^{M^{N_t}}{p(\mathbf{x}_p)e^{-b_{i,p}}d_{i,p}}}{\sum_{p=1}^{M^{N_t}}{p(\mathbf{x}_p)e^{-b_{i,p}}}}\right ],\\
		d_{i,p}&=\sigma^{-2}\mathbf{(\Sigma_H^2 \Sigma_G^2\Phi\Delta(x_\mathit{i}-x_\mathit{p})(x_\mathit{i}-x_\mathit{p})^\mathit{h}+\Sigma_H \Sigma_G v (x_\mathit{i}-x_\mathit{p})^\mathit{h}\Delta)}.
		\end{align}
	\end{subequations}
	\setcounter{equation}{19}
	\hrulefill
\end{figure*}

\subsection{Optimize signal distribution $p(\mathbf{x})$}
For given $\mathbf{\Sigma_G}$ and $\mathbf{\Phi}$, the subproblem is formulated as
	\begin{subequations}
	\begin{eqnarray}
	\underset{p(\mathbf{x})}{\max} &&\mathcal I\left(\mathbf x;\bar{\mathbf y}\right)\\
	\operatorname{s.t.}
	&&E_{\mathbf x}\left[\mathbf \Delta\mathbf x\left(\mathbf \Delta\mathbf x\right)^h\right]=\mathbf I_{N_t}.\label{constraintDelta}
	\end{eqnarray}
\end{subequations}

Under the premise that the signals from different antennas are independent of each other, we first introduce how the distribution of a single input stream is optimized, and then present the coordinate descent algorithm to optimize the distribution $p(\mathbf{x})$ alternatively.

For each antenna, the relation between constellation scaling $\Delta_j$ and the input distribution can be derived by maximizing the input entropy under the unit power constraint. To maximize the input entropy, the input signal should be subject to Maxwell-Boltzmann distributions \cite{MB} expressed as
\begin{equation}
p(x_i)=A e^{\lambda\left \| x_i \right \|^2 },A=\frac{1}{\sum_{i=1}^Me^{\lambda\left \| x_i \right \|^2}}.
\label{MB}
\end{equation}
To satisfy the power constrain $E_{\mathbf{x} }[\mathbf{\Delta x (\Delta x)}^h]=\mathbf{I}_{N_t}$, parameter $\lambda_j$ for $j$th antenna should satisfy the following equation
\begin{equation}
A_j\sum_{i=1}^Me^{\lambda_j\left \| x_i \right \|^2}\left \| x_i \right \|^2=\Delta_j^{-2},j=1, \dots,N_t.
\label{delta}
\end{equation}
For a given $\Delta_j$, the distribution of QAM symbols for $j$th antenna can be determined by \eqref{delta} through Newton's method. Thus, we can traverse the desirable range of $\Delta_j$ with a small step size to get an maximum gain on the mutual information. Based on the single stream optimization, we propose coordinate descent algorithm to optimize $p(\mathbf{x})$ for a vector input as shown in \textbf{Algorithm 3}.

\begin{algorithm}[tb]
\caption{The Coordinate Descent algorithm for optimizing $p(\mathbf{x})$ with given $\mathbf{\Sigma_H, \Sigma_G}$ and $\mathbf{\Phi}$.}
    \begin{algorithmic}[1]
    \State Initialize $k=0, \mathbf{\Delta}=\sqrt{\frac{3}{2(M-1)}}\mathbf{I}_{N_t}$. Compute the mutual information $\mathcal{I}_k$ with uniform input distribution. Compute the desirable range of constellation scaling: $\frac{1}{\sqrt{2}(\sqrt{M}-1)}\le \Delta_j \le \frac{1}{\sqrt{2}},j=1,\dots, N_t$. Decide the step size $\delta$ and threshold $\varepsilon $.
    \State For each $\Delta_j, j=1,\dots,N_t,$ traverse the feasible interval with step size $\delta$ and update $\Delta_j$ to the value that maximizes the mutual information. Compute the mutual information $\mathcal{I}_{k+1}$ after $\mathbf{\Delta}$ is updated, $u=\mathcal{I}_{k+1}-\mathcal{I}_k.$ $k:=k+1$.
    \State Repeat step 2 until $u\le \varepsilon$.
    \end{algorithmic}
\end{algorithm}

Based on the algorithms that optimize variables separately, we finally propose \textbf{Algorithm 4} to jointly optimize the mutual information $\mathcal{I}(\mathbf{x;\bar{y}})$ over the model \eqref{model_2}. Recall that the objective function
is monotonically non-increasing after each iteration of
\textbf{Algorithm 1-3}. Therefore, the proposed alternating
optimization algorithm is guaranteed to converge to a
suboptimal solution.

\begin{algorithm}[tb]
\caption{Joint optimization for $\mathcal{I}(\mathbf{x;\bar{y}})$}
    \begin{algorithmic}[1]
    \State For a given channel $\mathbf{H}$, apply SVD to get $\mathbf{\Sigma_H}$ and convert the original system model \eqref{model_1} into model \eqref{model_2}. Initialize $\mathbf{\Sigma_G}=\sqrt{\frac{P}{N_t}}\mathbf{I}_{N_t}, \mathbf{\Phi=I_{\mathit{N_t}}},\mathbf{\Delta}=\sqrt{\frac{3}{2(M-1)}}\mathbf{I}_{N_t}$, k=0 and compute $\mathcal{I}_k$. Compute the desirable range of constellation scaling: $\frac{1}{\sqrt{2}(\sqrt{M}-1)}\le \Delta_j \le \frac{1}{\sqrt{2}},j=1,\dots, N_t$. Decide the step size $\delta$ and threshold $\varepsilon $.
    \State Optimize $\mathbf{\Phi}$ according to \textbf{Algorithm 2}.
    \State Optimize $\mathbf{\Sigma_G}$ according to \textbf{Algorithm 1}.
    \State Optimize $p(\mathbf{x})$ according to \textbf{Algorithm 3}. Compute the mutual information $\mathcal{I}_{k+1}$ after $\mathbf{\Delta}$ is updated, $u=\mathcal{I}_{k+1}-\mathcal{I}_k.$ $k:=k+1$.
    \State Repeat step 2-4 until $u\le \varepsilon$.
    \end{algorithmic}
\end{algorithm}

\section{Simulation Results}
In this section, we evaluate our transmitting strategy in simulations, and compare it with the strategy proposed in \cite{Xiao}. The SNR is defined as $\mathrm{SNR=trace}(\mathbf{\Sigma_G}^2)/(N_r\sigma^2$). All simulations are performed under 16-QAM over $2\times2$ MIMO channels. As for massive MIMO, similar simulations can be done by blocking \cite{block}. Both mutual information results and FER performance are presented.

\subsection{Constant MIMO channel}
In this subsection, we consider a $2\times2$ MIMO system with the static channel matrix $\mathbf{H}=\begin{bmatrix}
  2&1 \\
  1&2
\end{bmatrix}$, whose singular value matrix can be normalized as $\mathbf{\Sigma_H}=\begin{bmatrix}
  1.3416&0 \\
  0&0.4472
\end{bmatrix}$. 

\begin{figure}[tbp]
	\centerline{\includegraphics[width=3.3in]{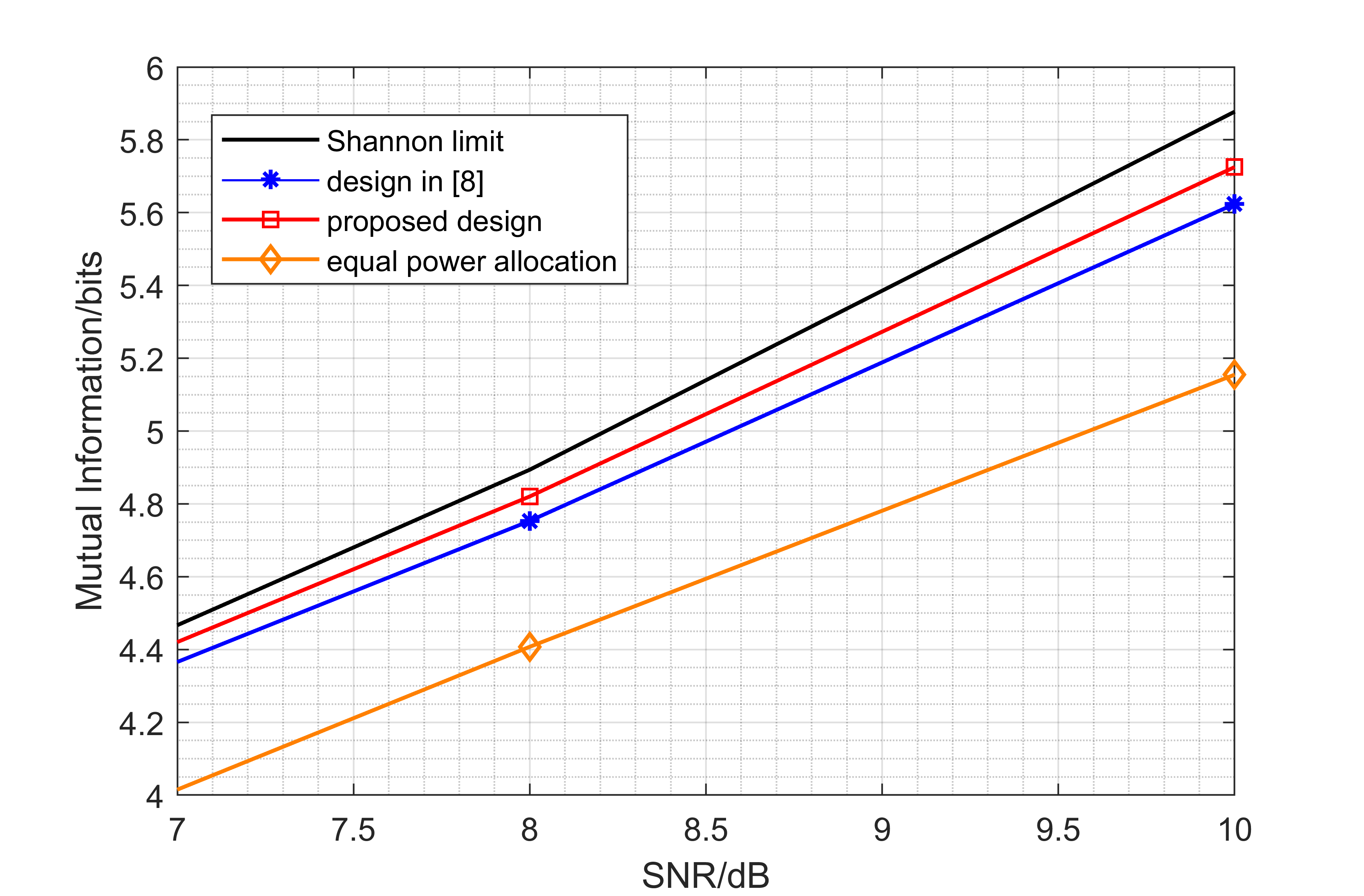}}
	\caption{Mutual information of the diagonal channel matrix $\mathbf{\Sigma_H}$ as a function of SNR with Gaussian and 16-QAM input.}
	\label{FigI}
\end{figure}
Fig. 1. plots the resulting values of mutual information of model \eqref{model_2}
as functions of the average SNR. Apart from our strategy, we also evaluate the Shannon Limit, the strategy applying optimal precoder proposed in \cite{Xiao}, and the strategy allocating equal power for parallel channels. It is illustrated that our strategy outperforms the existing optimal strategy \cite{Xiao}. Though optimizing the linear precoder already brings 1-2 dB gain in mid SNR compared with equal power allocation, the proposed design brings additional 0.2 dB SNR gain. And the gain will become higher for higher modulations \cite{higherMod}.

Fig. 2 shows the proportion of the transmitting power allocated to the stronger channel in different power allocation policies. We calculate the result of our proposed strategy and compare it with classic waterfilling\cite{Telatar}, mercury-waterfilling\cite{MWF} for uniformly distributed 16-QAM and equal power allocation. Contrary to mercury-waterfilling, the proposed design tends to allocate more power for the stronger channel while less for the weaker channel for $\mathrm{SNR>9dB}$. This is because the optimization of the the input distribution with PAS and unitary matrix $\mathbf{\Phi}$ makes the input signal more Gaussian distributed. Since the power allocation results in the proposed strategy are closer to classic waterfilling, we can apply classic waterfilling policy to initialize $\mathbf{\Sigma_G}$ in \textbf{Algorithm 1} for faster convergence.

\begin{figure}[tbp]
\centerline{\includegraphics[width=3.3in]{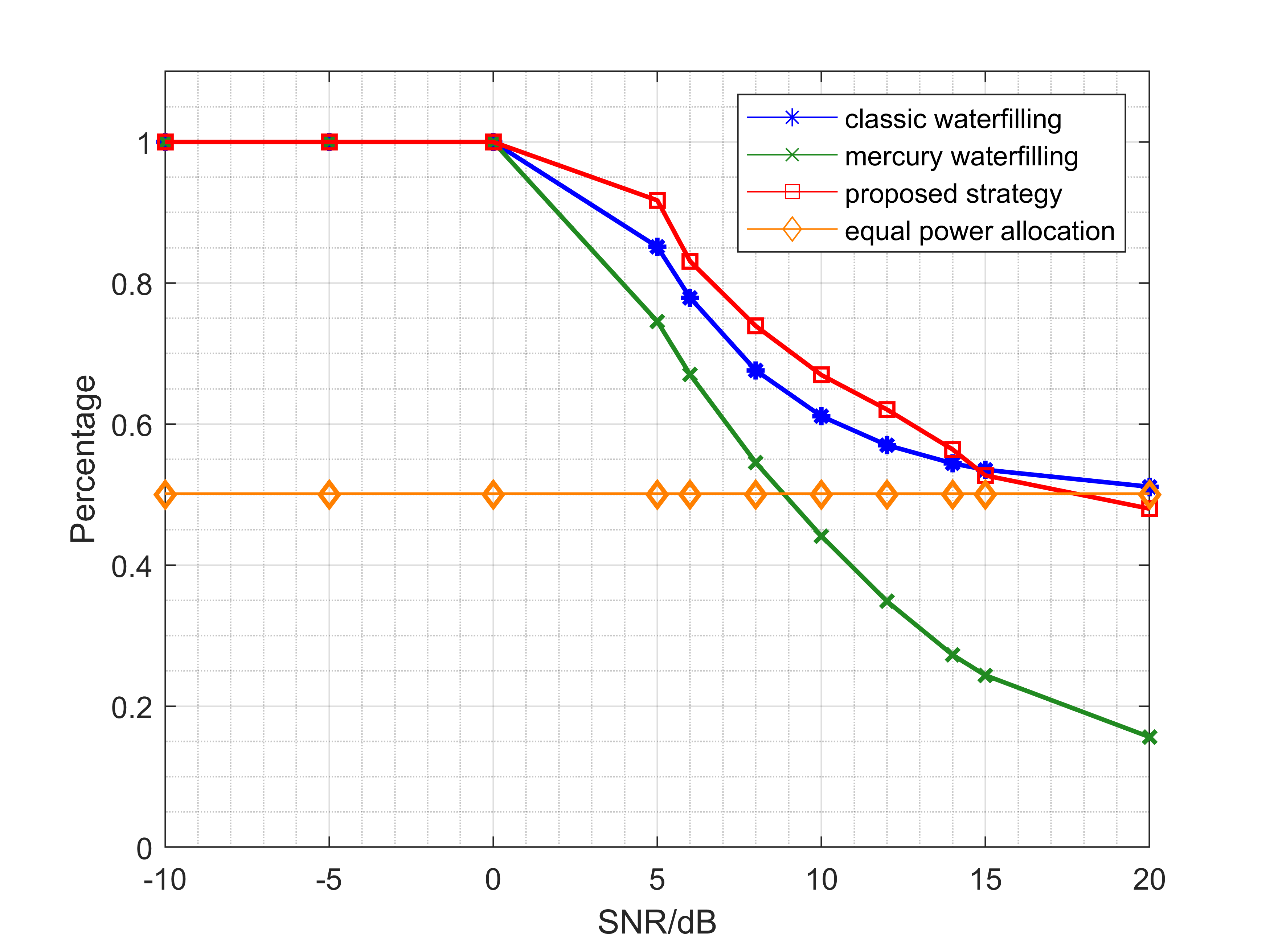}}
\caption{The percentage of the transmitting power allocated to the stronger channel of $\mathbf{\Sigma_H}$ as a function of SNR with 16QAM input .}
\label{FigI}
\end{figure}
\begin{figure}[tbp]
	\centerline{\includegraphics[width=3.3in]{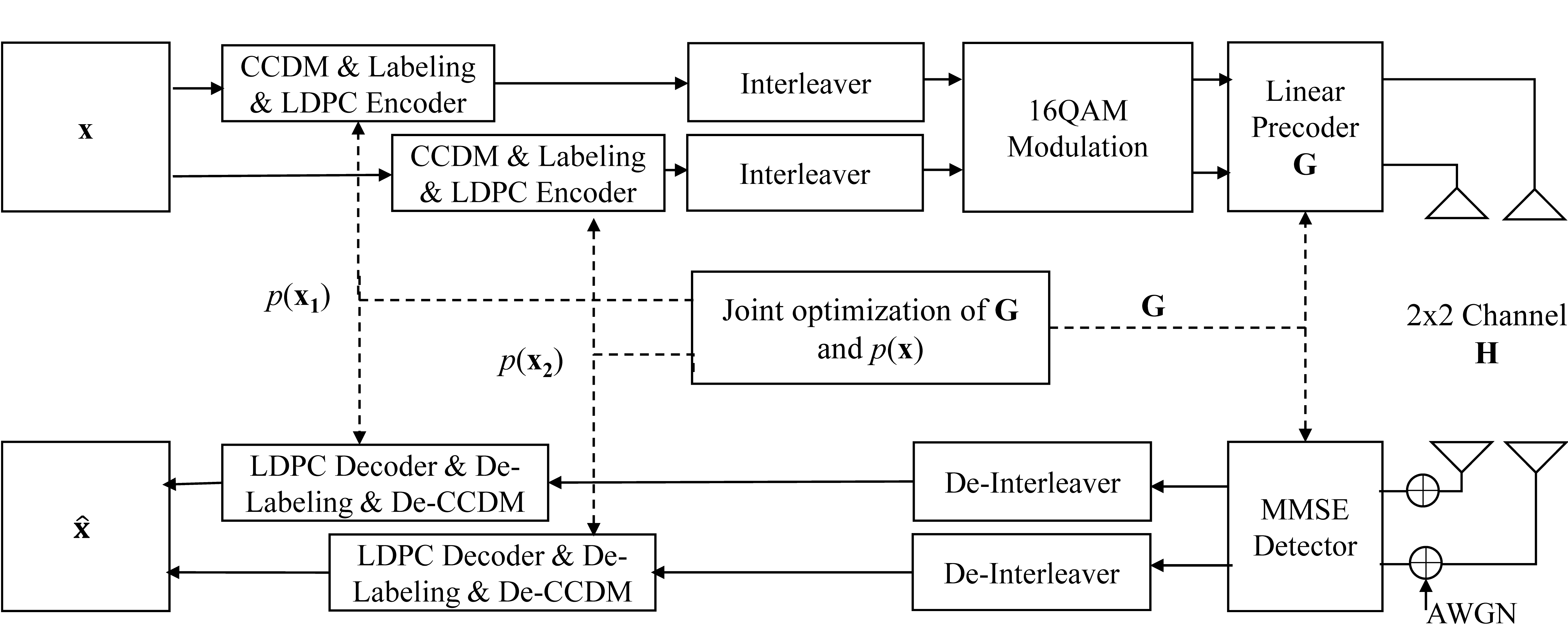}}
	\caption{The transceiver model of $2\times 2$ MIMO system with 16QAM input.}
	\label{FigI}
\end{figure}
\subsection{Rayleigh fading MIMO channel}
We also evaluate our design in $2\times2$ Rayleigh fading channels. The $2\times 2$ MIMO transceiver model for evaluating the FER performance with different transmitting strategies is displayed in Fig. 3. In our simulation, CCDM and DVS-S2 LDPC codes are combined to achieve 1/2 overall code rate transmission under 16QAM modulation. Although a capacity-achieving receiver is necessary to reach the Shannon limit, its high computational complexity makes it hardly used in practical. Thus, we simply apply the non-iterative MMSE detector at the receiver. Despite this compromise, the optimization of mutual information can still guide us to increase system capacity when designing transmission strategies.

The mutual information and FER curves respectively plotted in Fig. 4 and Fig. 5  are averaged from  $2\times2$ Rayleigh fading channels. We can observe from Fig. 4 that the proposed design achieves about 0.2 dB gain compared with the design in \cite{Xiao} and 1 dB gain compared with equal power allocation. And the proposed design even shows better performance in FER simulations. Fig. 5 elaborates that our proposed design brings about 1.2 dB and 3 dB gain comparing to the design in \cite{Xiao} and equal power allocation, respectively.

\begin{figure}[tbp]
	\centerline{\includegraphics[width=3.3in]{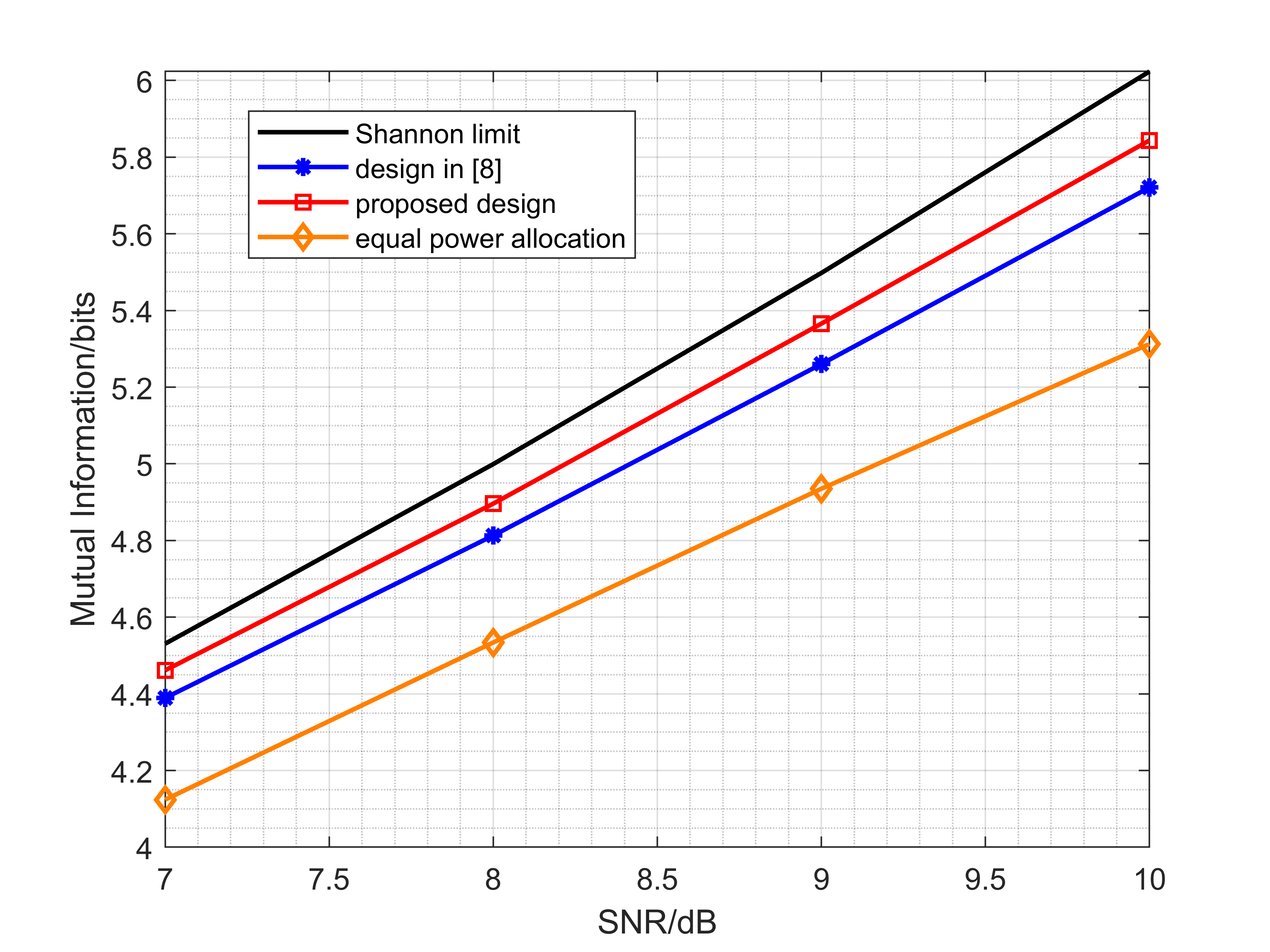}}
	\caption{The average mutual information of $2\times2$ Rayleigh fading channels as a function of SNR with Gaussian and 16-QAM input.}
	\label{FigI}
\end{figure}
\begin{figure}[tbp]
\centerline{\includegraphics[width=3.3in]{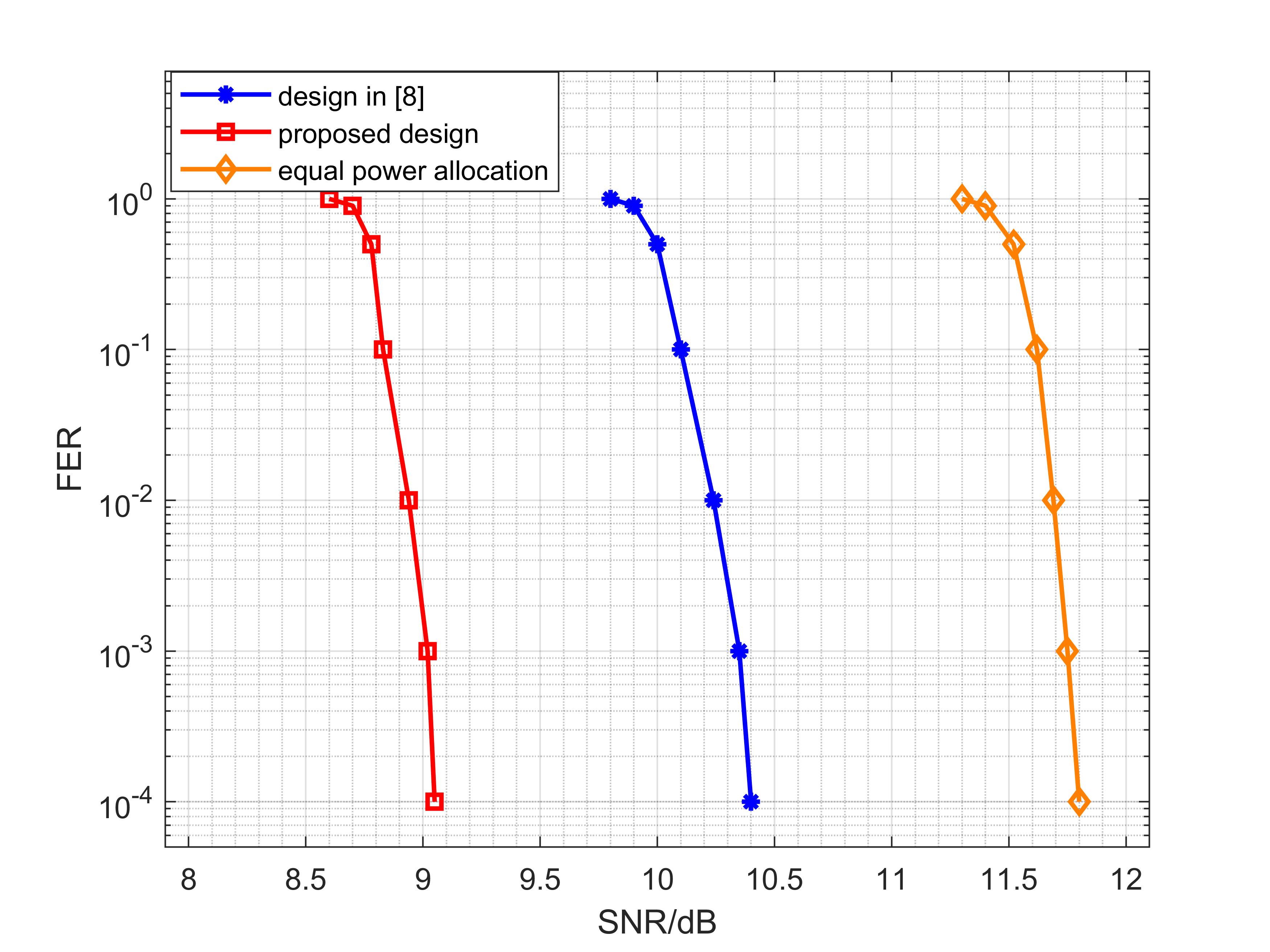}}
\caption{The average FER performance of Rayleigh fading $2\times2$ channels as a function of SNR with 16QAM input .}
\label{FigI}
\end{figure}

\vspace{1.5 ex}
\section{Conclusion}

In this paper, we study the MIMO transmission under discrete input signals constraints and present a unified framework with linear precoders and non-uniformly distributed input signals. First, the system model is defined. Based on this model, a joint optimization algorithm is proposed to solve the mutual information maximization problem. Next, We evaluate our strategy and other existing strategies for comparison. Numerical results indicate that the proposed strategy performs best. When evaluated over $2\times2$ Rayleigh fading channels with 16QAM input and the MMSE detector, our strategy achieves an SNR gain about 1.2 dB, comparing to the strategy only with the optimal precoder, averagely. In the future, we expect to evaluate the proposed strategy under more practical MIMO system constraints.

\end{document}